\documentclass[a4paper,11pt]{amsart}

\DeclareRobustCommand{\SkipTocEntry}[5]{}

\usepackage{ amsmath }
\usepackage{marvosym}

\usepackage[T1]{fontenc}
\usepackage{empheq}
\usepackage{relsize}
\usepackage{bm}

\usepackage{upgreek}
\usepackage{ esint }
\usepackage{color}
\usepackage{comment}
\usepackage{amssymb}
\usepackage{amsfonts}
\usepackage{graphicx}

\usepackage{amscd}

\usepackage{slashed}
\usepackage{pdflscape}
\usepackage{tikz}
\usepackage{enumerate}
\usepackage{multirow}
\usepackage{stmaryrd}

\usepackage[linkcolor=blue]{hyperref}
\usepackage{mathrsfs}
\usepackage[colorinlistoftodos]{todonotes}

\usepackage[T1]{fontenc}

\definecolor{blue}{rgb}{.255,.41,.884} 
\definecolor{red}{rgb}{1, 0, 0} 
\definecolor{green}{rgb}{.196,.804,.196} 
\definecolor{yellow}{rgb}{1,.648,0} 
\definecolor{pink}{rgb}{1,0.5,0.5}

\newtheorem{theorem}{Theorem}[section]
\newtheorem{lemma}[theorem]{Lemma}

\theoremstyle{definition}
\newtheorem{definition}[theorem]{Definition}
\newtheorem{example}[theorem]{Example}

\theoremstyle{remark}
\newtheorem{remark}[theorem]{Remark}

\usepackage{pdfsync}

\usepackage{graphicx} 
\usepackage{amssymb}

\newcommand{\be}{\begin{equation}}

\newcommand{\ee}{\end{equation}}

\newcommand{\ba}{\begin{array}}

\newcommand{\ea}{\end{array}}

\newcommand{\beq}{\begin{eqnarray}}

\newcommand{\eeq}{\end{eqnarray}}

\newtheorem{lm}{lemma}

\newtheorem{thee}{theorem}

\newtheorem{proo}{proposition}

\newtheorem{co}{corollary}

\newtheorem{rem}{remark}

\newtheorem{deff}{definition}

\newcommand{\bd}{\begin{deff}}

\newcommand{\ed}{\end{deff}}

\newcommand{\bl}{\begin{lm}}

\newcommand{\el}{\end{lm}}

\newcommand{\bp}{\begin{proo}}

\newcommand{\ep}{\end{proo}}

\newcommand{\bt}{\begin{thee}}

\newcommand{\et}{\end{thee}}

\newcommand{\bc}{\begin{co}}

\newcommand{\ec}{\end{co}}

\newcommand{\brm}{\begin{rem}}

\newcommand{\erm}{\end{rem}}

\usepackage{t1enc}
\def\frak{\mathfrak}

\newcommand{\newc}{\newcommand}

\renewcommand{\exp}{\operatorname{exp}}

\let\ccdot.

\newc{\aR}{\mbox{\boldmath{$ R$}}}
\newc{\aS}{\mbox{\boldmath{$ S$}}}
\newc{\aT}{\mbox{\boldmath{$ T$}}}
\newc{\aW}{\mbox{\boldmath{$ W$}}}

\newc{\aD}{\mbox{\boldmath{$ D$}}\hspace{-.2mm}}

\newc{\aK}{\mbox{\boldmath{$ K$}}}
\newc{\aL}{\mbox{\boldmath{$ L$}}}

\usepackage{amssymb}
\usepackage{amscd}

\newcommand{\nn}[1]{(\ref{#1})}

\newc{\obstrn}[2]{B^{#1}_{#2}}

\newcommand{\rpl}                         
{\mbox{$
\begin{picture}(12.7,8)(-.5,-1)
\put(0,0.2){$+$}
\put(4.2,2.8){\oval(8,8)[r]}
\end{picture}$}}

\newcommand{\lpl}                         
{\mbox{$
\begin{picture}(12.7,8)(-.5,-1)
\put(2,0.2){$+$}
\put(6.2,2.8){\oval(8,8)[l]}
\end{picture}$}}

\usepackage{ifthen}

\newc{\tensor}[1]{#1}
\newc{\Mvariable}[1]{\mbox{#1}}
\newc{\down}[1]{{}_{#1}}
\newc{\up}[1]{{}^{#1}}

\newc{\JulyStrut}{\rule{0mm}{6mm}}
\newc{\midtenPan}{\mbox{\sf S}}
\newc{\midten}{\mbox{\sf T}}
\newc{\midtenEi}{\mbox{\sf U}}
\newc{\ATen}{\mbox{\sf E}}
\newc{\BTen}{\mbox{\sf F}}
\newc{\CTen}{\mbox{\sf G}}

\def\sideremark#1{\ifvmode\leavevmode\fi\vadjust{\vbox to0pt{\vss
 \hbox to 0pt{\hskip\hsize\hskip1em
 \vbox{\hsize2cm\tiny\raggedright\pretolerance10000
  \noindent #1\hfill}\hss}\vbox to8pt{\vfil}\vss}}}

\numberwithin{equation}{section}

\newcommand{\hh}{{\hspace{.3mm}}}

\newcommand{\sss}{\scriptscriptstyle}

\renewcommand\geq{\geqslant}
\renewcommand\leq{\leqslant}

\newcommand{\Partial}[1]{\frac{\partial}{\partial #1}}

 \newcommand{\bdot }{{\mathop{\lower0.33ex\hbox{\LARGE$\cdot$}}}}

\newcommand{\superimpose}[2]{%
  {\ooalign{$#1\@firstoftwo#2$\cr\hfil$#1\@secondoftwo#2$\hfil\cr}}}

\begin{document}

\renewcommand{\today}{}
\title{
{
Contact Quantization: \\[1mm] {\small Quantum Mechanics = Parallel transport 
 }}}
\author{ G.~Herczeg${}^\sharp$, E.~Latini${}^\flat$ \&  Andrew Waldron${}^\natural$}

\address{${}^\sharp$
Department of  Physics,
  University of California,
  Davis, CA 95616, USA}\email{Herczeg@ms.physics.ucdavis.edu}

\address{${}^\flat$ 
 Dipartimento di Matematica, Universit\`a di Bologna, Piazza di Porta S. Donato 5,
 and  INFN, Sezione di Bologna, Via Irnerio 46, I-40126 Bologna,  Italy}
 \email{emanuele.latini@UniBo.it}

  \address{${}^{\natural}$
  Center for Quantum Mathematics and Physics (QMAP)\\
  Department of Mathematics\\ 
  University of California\\
  Davis, CA95616, USA} \email{wally@math.ucdavis.edu}

\vspace{10pt}

\renewcommand{\arraystretch}{1}

%
%
%
%

%
%
%
%
%
%


\begin{abstract}
\noindent

Quantization together with quantum dynamics can be simultaneously formulated as the problem of finding an appropriate flat connection on a Hilbert bundle over a contact manifold. 
Contact geometry treats time, generalized positions and momenta as points on an underlying phase-spacetime and reduces classical mechanics to contact topology.
{\it Contact quantization} describes quantum dynamics in terms of parallel transport for a flat connection; the ultimate goal being to also handle quantum systems in terms of contact topology. 
Our main result is a  proof of local, formal gauge equivalence for a broad class of quantum dynamical systems---just as classical dynamics depends on choices of clocks, local quantum dynamics 
can be reduced to a problem of studying gauge transformations. We
further 
show how to write quantum correlators in terms of parallel transport and in turn matrix elements for Hilbert bundle gauge transformations,
and give the path integral formulation of these results. Finally, we show how to relate topology of the underlying contact manifold to boundary conditions for quantum wave functions.

 \end{abstract}


\maketitle

\pagestyle{myheadings} \markboth{Herczeg, Latini \& Waldron}{Contact Quantization}


\tableofcontents

\section{Introduction}

To understand why a study of contact geometry is fundamental to quantum mechanics, it is useful to think about the standard Copenhagen intepretation 
in a novel way: According to the  Copenhagen interpretation, one prepares an initial quantum state, allows it to evolve for some time, and then calculates the probability of observing some choice of final state.
The basic data here is a Hilbert space
and a one parameter family of unitary operators that determine time evolution. 
This parameter typically corresponds to time intervals as measured in a classical laboratory. Two modifications of this standard paradigm will lead us to a---rather propitious---reformulation of quantum mechanics as a theory of flat connections on a Hilbert bundle over a contact manifold: 

\begin{enumerate}[(i)]
\item
Because it ought be possible to describe quantum dynamics for any choice of laboratory time coordinate (for example one may conceive of notions of time that mix varying combinations of classical-laboratory measurements), we replace the time interval with a classical ``phase-spacetime'' manifold~$Z$, which can be thought of as a classical phase space augmented by a timelike direction that enjoys general coordinate (diffeomorphism) invariance.
\medskip
\item
Instead of viewing quantum dynamics as trajectories in a single given Hilbert space~${\mathcal H}$,
 we associate---in a manner reminiscent of  gauge theories and general relativity---a copy of the Hilbert space to every point in the phase-spacetime $Z$. This structure is a {\it Hilbert bundle}~$Z\ltimes{\mathcal H}$, {\it viz.} a vector bundle whose fibers are Hilbert spaces~\cite{Dupre}. We use the warped product notation $Z\ltimes{\mathcal H}$ to indicate that, locally in $Z$, the Hilbert bundle is a direct product, although this need not globally be the case.
\end{enumerate}

Given the  geometric data of the vector bundle $Z\ltimes {\mathcal H}$, we wish to compare Hilbert space states at distinct points in $Z$. For that we need a connection $\nabla$. Concretely 
$$
\nabla = d + \widehat A\, ,
$$
where $d$ is the exterior derivative on $Z$ and $i\widehat A$ is a one-form taking values in the space of hermitean operators on ${\mathcal H}$. In particular, if ${\mathcal H}$ is simply $L^2({\mathbb R}^n)$, we may consider $\widehat A$ to take values in the self-adjoint subspace of the corresponding Weyl algebra. 

To construct the connection $\nabla$, additional data is required. In Section~\ref{II}, we will show that giving the phase-spacetime manifold a strict contact stucture
endows the Hilbert bundle $Z\ltimes{\mathcal H}$ with a flat connection. Physically, this strict contact data corresponds to specifying classical dynamics on  $Z$. The construction we give is partly motivated by earlier BRST studies of Fedosov quantization~\cite{Fed} for symplectic manifolds~\cite{BRSTFed}.
Solutions to the  quantum Schr\"odinger equation are then 
 parallel sections of the Hilbert bundle---quantum dynamics amounts to parallel transport of states from one Hilbert space fiber to another. The main theorem of 
of Section~\ref{II} establishes that solutions for connections obeying the flatness condition are locally and formally gauge equivalent. The method of proof is close to that employed in Fedosov's original work on deformation quantization of Poisson structures~\cite{Fed}. The key advantage is that our contact approach not only incorporates dynamics, but also establishes a very general local gauge equivalence between dynamical quantum systems.

In Section~\ref{dynamics}, we 
focus on the description of dynamics in terms of parallel sections of the Hilbert bundle. 
In particular we show how to reduce the problem of computing quantum correlators to that of finding the matrix element of a gauge transformation. We also give a path integral description of correlators in terms of  paths in  a novel extended phase-spacetime description of contact Reeb dynamics. We also show how topology of the underlying contact manifold determines boundary conditions for quantum wavefunctions.
 Open problems and future prospects are discussed in Section~\ref{discuss}.

\section{Strict Contact Structures and Quantization}\label{II}

Contact geometry may be viewed as a unification of Hamiltonian dynamics and symplectic geometry. Therefore, before discussing quantization, we introduce the salient features of contact structures~\cite{Geiges,Rajeev}.

\subsection{Contact geometry}

A {\it strict contact structure} is the data $(Z,\alpha)$ where $Z$ is a  $2n+1$ dimensional manifold and $\alpha$ is a {\it contact one-form}, meaning that the volume form
 \begin{equation}
{\rm Vol}_\alpha :=\alpha\wedge \varphi^{\wedge n}
\end{equation}
is nowhere vanishing\footnote{A {\it contact structure} is the data of a maximally non-integrable hyperplane distribution; the kernel of~$\alpha$ (viewed as a map on tangent spaces $T_PZ\to {\mathbb R}$) determines precisely such a distribution (as does any $f \alpha$ where $0<f\in C^\infty Z$). Note also, that it is  interesting to consider models for which the Levi-form $\varphi= d\alpha$ has maximal rank, but ${\rm Vol}_\alpha$ may vanish (either locally or globally). The massless relativistic particle falls into this class.}, where the two form
$$
\varphi:=d\alpha\, ,
$$
determines the {\it Levi-form} along the distribution; we therefore also term $\varphi$ the Levi.
\medskip

The data $(Z,\alpha)$ allows us to formulate classical dynamics via the action principle
\begin{equation}
S = \int_{\gamma}\alpha\, ,
\label{action}
\end{equation}
defined by integrating the contact one-form along {\it unparameterized} paths $\gamma$ in $Z$. Requiring $S$ to be extremal under compact variations of the embedding $\gamma \hookrightarrow Z$ yields equations of motion
\begin{equation}
\label{Reeb}
\varphi(\dot \gamma,\bdot)=0\, .
\end{equation}
Since the Levi-form necessarily has maximal rank, the above condition determines the tangent vector to $\gamma$ up to an overall scale. The choice of solution $\dot \gamma=\rho$ to Equation~\nn{Reeb} with normalization
$\alpha(\rho)=1$ is called the {\it Reeb vector}.  Classical evolution is therefore governed by flows of the Reeb vector; and in this context is dubbed {\it Reeb dynamics}. It is not difficult to verify that these obey a contact analog of the classical Liouville theorem, namely that the volume form is preserved by Reeb dynamics:
$$
{\mathcal L}_\rho {\rm Vol}_\alpha= 0\, ,
$$
where ${\mathcal L}_\bdot$ denotes the Lie derivative.

The contact Darboux theorem is particularly powerful; it ensures that locally there exists a diffeomorphism on~$Z$ that brings any contact form to the normal form
\begin{equation}\label{normal}
\alpha = \uppi_A d\upchi^A - d\uppsi\, ,
\end{equation}
where $(\uppi_A, \upchi^A, \uppsi)$ are $2n+1$ local coordinates for $Z$.
On this coordinate patch the Reeb vector~$\rho=-\frac{\partial}{\partial \uppsi}$
so that dynamics are locally trivial.
Observe that in the worldline diffeomorphism  gauge $\uppsi=\tau$, where $\tau$ is a worldline parameter along $\gamma$, the action~\nn{action} becomes
$$
S=\int d\tau \big[\uppi_a \dot \upchi^a -1\big]\, .
$$
This is the Hamiltonian action principle for  a system with Darboux symplectic form $d\uppi_a \wedge d\upchi^a$ and trivial Hamiltonian $H=1$.

\subsection{Constraint analysis}

Our quantum BRST treatment of  Reeb dynamics requires that  we  examine the constraint structure of the model~\nn{action}. Firstly observe that
the action principle~\nn{action} is worldline diffeomorphism invariant, and in a choice of coordinates~$z^i$ for $Z$  reads  $S=\int\alpha_i(z) \dot{z}^i d\tau$. Therefore the canonical momenta~$p_i$ for $\dot z^i$ obey $2n+1$ constraints
$$
C_i:=p_i - \alpha_i(z)=0\, ,
$$ 
of which $2n$ are second class (because these constraints Poisson commute to give the maximal rank Levi-form: $\{C_i,C_j\}_{\rm PB}=\varphi_{ij}$) and one is first class (corresponding to worldline diffeomorphisms). By introducing $2n$ ``fiber coordinates'' $s^a$ (see~\cite{BFVsecond}), local classical dynamics can be described by an equivalent extended action principle 
for paths~$\Gamma$ in ${\mathcal Z}:=Z\times{\mathbb R}^{2n}$ for
which all constraints are first class\footnote{To analyze global dynamics  one ought promote ${\mathcal Z}$ to a bundle $Z\ltimes {\mathbb R}^{2n}$.}:
\begin{equation}
\label{extS}
S_{\rm \tiny ext}=\int_\Gamma\Big[
\tfrac 12 s^a J_{ab} d s^b  + A(s)
\Big]\, .
\end{equation}
In the above $J_{ab}$ is a constant, maximal rank antisymmetric matrix (and therefore an invariant tensor for the Lie algebra $\mathfrak{sp}(2n)$).
The one-form $A$ is given by
$$
A(s)= \alpha + e^a J_{ab} s^b + \omega(s)\, ,
$$
where the {\it soldering forms} $e^a$ together with the contact one-form $\alpha$ are a basis for $T^*Z$ such that the Levi-form decomposes as
$$
\varphi =\frac 12\,  J_{ab} e^a\wedge e^b\, ,
$$
and $e^a(\rho)=0$. The extended action~\nn{extS} enjoys $2n+1$ gauge invariances (and hence ${2n+1}$, abelian, first class constraints) when  $A$ obeys the zero curvature type condition\footnote{For a pair of one-forms $A$ and $B$, we denote $\{A(s)\wedge B(s)\}_{\rm PB}:=J^{ab} \frac{\partial A}{\partial s^a} \wedge  \frac{\partial B}{\partial s^b}$  where the inverse matrix $J_{ab}$ obeys $J_{ab}J^{bc}=\delta^c_a$.} $$
dA + \frac12 \{A\wedge A\hh \}_{\rm PB}=0\, .
$$
This condition can be used to  determine the one-form $\omega(s)$ to any order  in a formal power series in~$s$ (and therefore exactly for contact forms expressible as polynomials in some coordinate system).  The main ingredients for quantization are now ready.

\subsection{Flat connections}
Because the constraints are now abelian and first class, it is straightforward to  quantize the {\it extended Reeb dynamics} defined by the action~\nn{extS} using the Hamiltonian BRST technology of~\cite{BFV}. The resultant nilpotent BRST charge 
may be interpreted as a flat connection~$\nabla$  on the Hilbert bundle $Z\ltimes{\mathcal H}$. [An analogous connection has been constructed for symplectic manifolds in~\cite{Krysl}.] 
In detail, 
$$
\nabla = d +\widehat A\, ,
$$
where $\widehat A$ is a one-form taking hermitean values in the enveloping algebra ${\mathcal U}({\mathfrak {heis}})$ of the Heisenberg algebra 
\begin{equation}\label{Heisalg}
{\mathfrak {heis}}={\rm span}\{1,\hat s^a\}\, ,\qquad [\hat s^a,\hat s^b]=i\hbar J^{ab}\, .
\end{equation}
In particular
$$
i\widehat A =  \, \frac{\alpha}{\hbar}  + \frac{e^a J_{ab} \hat s^b}{\hbar} + i\, \widehat \Omega\, ,
$$
where $\hbar \widehat \Omega$ is a hermitean operator,  potentially involving higher powers of the generators~$\hat s^a$, that is  expressible as a formal power series in $\hbar$.
It is formally determined  by the zero curvature condition
\begin{equation}\label{zipcurves}
\nabla^2=0\, .
\end{equation}

\medskip

\begin{example}[Hamiltonian dynamics]
\label{HamD}

Let $Z={\mathbb R}^3=\{p,q,t\}$ and $$\alpha= pdq - H(p,q,t) dt\, ,$$ 
with Hamiltonian $H$ given by a (possibly time-dependent) polynomial in $p$ and $q$. Notice that
$
\varphi = e \wedge f$ where $e:= dp +\frac{\partial H}{\partial q} dt$ and $f:=dq-\frac{\partial H}{\partial p} dt
$, so we make a choice of soldering $e^a=(f,e)$ which we use to construct the flat connection: 
\begin{equation}\label{Hnabla}
\nabla = d + \frac i \hbar \Big[dp\,  S - dq \, \Big(p+\frac \hbar i \frac\partial{\partial S}\Big)\Big]
+\frac i\hbar \, dt  \widehat H\, ,
\end{equation}
where the operator
$$
\widehat H := \Big(H\big(q+S,p+\frac\hbar i \frac\partial{\partial S}\big)\Big)_{\rm Weyl}
$$
is given by Weyl ordering the operators\footnote{Note that we have made the choice of Hilbert space ${\mathcal H}=L^2({\mathbb R})$ here as well as a polarization for the space of wavefunctions. Different choices of polarization differ only by gauge transformations---recall that in its metaplectic representation, compact elements of $sp(2n)$ act by Fourier transform on Schwartz functions.} $\hat s^a := (S,\frac\hbar i \frac\partial{\partial S})$ (This ensures formal self-adjointness of the operator $\widehat H$.) The Schr\"odinger equation~\nn{Schroedinger} may be solved by setting $\Psi=\exp(-\frac i\hbar p S) \, \psi(q+S,t)$, where $\psi(Q,t)$ obeys the standard time dependent Schr\"odinger equation
$$
i\hbar \hh\frac{\partial \psi(Q,t)}{\partial t} = \Big(H\big(Q,\frac\hbar i\frac\partial{\partial Q}\big)\Big)_{\rm \tiny Weyl}\hh  \psi(Q,t)\, .
$$
This example therefore shows how contact quantization  recovers standard quantum mechanics.
\end{example}

\medskip

To better understand the space of flat connections $\nabla$, 
we further organize the expansion in powers of operators $\hat s$  by assigning a grading $\sf gr$ to the operators $\hat s$ and $\hbar$ where\footnote{When applied to sums of terms inhomogeneous in the grading, we define ${\sf gr}$ by the grade of the lowest grade term. }
$$
{\sf gr}(\hbar) =2\,, \qquad
{\sf gr}(\hat s^a) =1\, .
$$
Thus, arranging the connection in terms of this grading we have 
$$
\nabla = \underbrace{\frac{\alpha}{i\hbar}}_{-2}  + 
\underbrace{\frac{e^a J_{ab} \hat s^b}{i\hbar} }_{-1}
+\underbrace{\, d_{w_{\phantom{A_A\!\!\!\!\!\!}}}}_{0}+ \underbrace{\widehat\omega}_{\geq 1}\, ,
$$
where
$$
d_\omega:= d + \frac1{2i\hbar}\, \omega_{ab}\hat s^a \hat s^b\, .
$$
Here the symmetric part of $\omega_{ab}$ gives an $\frak{sp}(2n)$-valued one-form (or connection) while the antisymmetric part is necessarily pure imaginary in order that $\widehat \Omega$ is hermitean.
Also, the terms with strictly positive grading  are $\widehat \omega:=\widehat\Omega-\frac1{2i\hbar}\, \omega_{ab}\hat s^a \hat s^b$.

Observe that this grading is invariant under rewritings of products of the operators $\hat s$ given by quantum reorderings, for example 
$$
{\sf gr}(\hat s^a \hat s^b) = {\sf gr}\Big(\hat s^b \hat s^a +i\hbar  J^{ab}\Big)\, .
$$
In other words, ${\sf gr}$ filters
 ${\mathcal U}({\mathfrak {heis}})$. The projection of an element in ${\mathcal U}({\mathfrak {heis}})$
 to the part of grade $k$ is denoted by\footnote{We also employ ${\sf gr_K}(\bdot)$, where $K\subset {\mathbb Z}$, to denote projection to 
 subspaces with the corresponding grades. For the exterior derivative, we define ${\sf gr}(d)=0$.
 } ${\sf gr_k}(\bdot)$.
 
In Theorem~\ref{gaugeflat}  we shall show that locally, every solution to the flatness condition~\ref{zipcurves} is formally\footnote{The terms {\it formally equivalent}  here are defined to mean that  gauge transformations exist giving connections that are equal to any chosen order in the grading ${\sf gr}$.} gauge equivalent\footnote{To be sure, we are not claiming that this means all quantum dynamics on a given Hilbert space are equivalent, rather
having identified the physical meaning of variables for a given connection $\nabla$, the ``gauge equivalent'' (in the bundle sense) connection $\nabla'=\widehat U \, \nabla \, \widehat U^\dagger$ will in general describe different dynamics.
  This is much like the case of active diffeomorphisms for a theory in a fixed
 generally curved background.
Moreover, it is a highly useful feature, because at least locally, it allows complicated dynamics to be described in terms of simpler ones.
 } to a connection where $\widehat \Omega=0$. Moreover the latter such solutions always exist.

 \color{red}
\color{black}

\medskip

Realizing $\hat s^a$ by hermitean operators representing the Heisenberg algebra acting on~$\mathcal H$, the (principal) connection $\nabla$ gives a connection on the (associated) Hilbert bundle $Z\ltimes {\mathcal H}$. The Schr\"odinger equation 
is then simply the parallel transport condition
\begin{equation}\label{Schroedinger}
\nabla \Psi = 0
\end{equation}
on Hilbert bundle sections $\Psi\in \Gamma(Z\ltimes {\mathcal H})$. Indeed, modulo (non-trivial) global issues, the problem of quantizing a given classical system now amounts to solving the above flat connection problem~\nn{zipcurves}, while quantum dynamics amounts to parallel transport.

\medskip

\begin{theorem}\label{gaugeflat}
Any two flat connections $\nabla=d+\widehat A$ and $\nabla'=d+\widehat A'$ where
$$
{\sf gr}_{-2}(\widehat A\hh )=\frac{\alpha}{i\hbar}={\sf gr}_{-2}(\widehat A')\, ,
$$
are locally, formally gauge equivalent.
\end{theorem}

\begin{proof}
The contact Darboux theorem ensures that locally, there exists a set of \emph{closed} one-forms $dE^a = 0$, such that $$
 \varphi =\tfrac{1}{2}
 J_{ab}E^a\wedge E^b  \mbox{ and } \iota_\rho E^a=0\, .$$
(In the normal form~\nn{normal}, $E^a=(d\upchi^A,d\uppi_A)$.)
Hence the connection
\begin{equation}\label{dbc}
\nabla_{\rm D}:=\frac{\alpha}{i\hbar}+\frac{E^a J_{ab} \hat s^b}{i\hbar} + d
\end{equation}
solves the flatness condition~\nn{zipcurves}.
Our strategy is to construct the gauge transformation bringing a general flat $\nabla$ to this ``Darboux form''.

Firstly, the flatness condition of a general $\nabla=d+\widehat A$ at grade $-2$ implies that 
$$
\frac{d\alpha}{i\hbar } + \Big({\sf gr}_{-1}\big( \widehat A\, \big)\Big)^2=0\, .
$$
This is solved, as discussed earlier, by
$$
i\hbar \, {\sf gr}_{-1} \widehat A= e^a J_{ab} \hat s^b\, ,
$$
where
$$
\varphi = \tfrac{1}{2}J_{ab}e^a\wedge e^b\mbox{ and } \iota_\rho e^a=0\, .
$$
Comparing the line above with the first display of this proof, we see there must (pointwise in some neighborhood in $Z$) exist an invertible linear transformation $U\in GL(2n)$ such that
$$E^a = U^a{}_{b}e^b\, .$$
Moreover, $U$ must preserve $J$ and hence is in fact $Sp(2n)$-valued with unit determinant. Thus, we may write $U=\exp(u)$. In turn it follows that
$$
{\sf gr}_{\{-2,-1\}}\big(\exp(\hat u_0) \widehat A \exp(-\hat u_0 )\big)=
\frac{\alpha}{i\hbar}+
\frac{
E^a J_{ab} \hat s^b}{i\hbar}\, , 
$$
where 
$$\hat u_0=\frac{
J_{ac}u^c{}_b
\hat s^a \hat s^b}{2i\hbar}\, .$$ 
Essentially, we have just intertwined $U$ in the fundamental representation of $Sp(2n)$ to its metaplectic representation.

We now observe that
\begin{equation}\label{LOOKHERE}
{\sf gr}_{0}\big(\exp(\hat u_0) (d+\widehat A\, ) \exp(-\hat u_0 )\big)=
d
-i
\alpha_1+
\frac{\omega_{ab}
\hat s^a \hat s^b}{2i\hbar}\, , 
\end{equation}
where $\alpha_1$ is some real-valued, $\hbar$-independent one-form and the one-form $\omega_{ab}=\omega_{ba}$
(the Heisenberg algebra~\nn{Heisalg} may be used to absorb an antisymmetric part of $\omega_{ab}$ in $\alpha_1$).

We now again employ flatness of $\nabla$ and closedness of the $E^a$'s to obtain 
$$
0=
{\sf gr}_{-1}\Big(\big(\exp(\hat u_0) (d+\widehat A\, ) \exp(-\hat u_0 )\big)^2\Big)=\frac{\omega_{ab}\wedge E^a\, \hat s^b}{i\hbar}\, .
$$
We decompose the one-form $\omega^{ab}$ with respect to the (local) basis $(\alpha,e^a)$ for $T^*Z$ as
$
\omega_{ab}=W_{ab}\,  \alpha + W_{abc} E^c
$. 
The above display then implies that the functions $W_{ab}$ must vanish and 
$$
W_{abc}E^a \wedge E^c=0\, .
$$
Hence $W_{abc}$ is totally symmetric in the indices $a,b,c$.

We now gauge away the term $\omega_{ab}\hat s^a \hat s^b/(2i\hbar)=W_{abc}\hat s^a \hat s^b E^c/(2i\hbar)$ in Equation~\nn{LOOKHERE}.
Since we are working formally order by order in the grading, we may employ the Baker--Campbell--Hausforff formula $\exp(\hat u)  \, \widehat W \exp(-\hat u ) = \exp([\hat u,\bdot])(\widehat W)$. In particular
$$
{\sf gr}_{0}\Big(\exp(\hat u_1)  \frac{E^a J_{ab}\hat s^b}{i\hbar} \exp(-\hat u_1 )\Big)=-\frac{W_{abc}\hat s^a \hat s^b E^c}{2i\hbar}\, ,
$$
for the choice $\hat u_1= W_{abc} \hat s^a\hat s^b \hat s^c/(3!i\hbar)$. 
Hence we have achieved
$$
{\sf gr}_{\{-2,-1,0\}}\Big(
\exp(\hat u_1)\exp(\hat u_0) (d+\widehat A\, ) \exp(-\hat u_0 )
\exp(-\hat u_1 )
\Big)
=\frac{\alpha}{i\hbar}
+\frac{E^a J_{ab}\hat s^b}{i\hbar}
+d-i\alpha_1\, .
$$
 
 At this juncture, we have established the base case for an induction.
 Proceeding recursively we now assume that the flat connection $\nabla=d+\widehat A$ obeys
 $$
 {\sf gr}_{\{-2,\ldots,k\}}(\widehat A\, )=\frac{\alpha+\hbar \alpha_1+\cdots+\hbar^{[(k+1)/2]}\alpha_{[(k+1)/2]}}
 {i\hbar}+\frac{E^a J_{ab}\hat s^b}{i\hbar}
+d+\hat \omega_k\, ,
 $$
 where $\alpha_i$ are $\hbar$-independent one-forms and, without loss of generality, take  ${\sf gr}(\hat \omega_k)= k$.
 
Employing the flatness condition for $\nabla$ along the same lines explained above to $\hat \omega_k$ shows that
$$
i\hbar \hat \omega_k =
\left\{
\begin{array}{r}
\frac{1}{(k+2)!}\, 
W_{a_1\ldots a_{k+3}}\hat s^{a_1}\cdots \hat s^{a_{k+2}} E^{a_{k+3}} + \frac{\hbar}{k!}\, W_{a_1\ldots a_{k+1}}\hat s^{a_1}\cdots \hat s^{a_{k}} E^{a_{k+1}} +
\cdots\qquad \\[2mm]
+\, {\hbar^{(k+1)/2}}W_{a_1a_2}\hat s^{a_1}E^{a_2}\, ,\quad k \mbox{ odd}\, ,\\[3mm]
\frac{1}{(k+2)!}\, 
W_{a_1\ldots a_{k+3}}\hat s^{a_1}\cdots \hat s^{a_{k+2}} E^{a_{k+3}} + \frac{\hbar}{k!}\, W_{a_1\ldots a_{k+1}}\hat s^{a_1}\cdots \hat s^{a_{k}} E^{a_{k+1}} +
\cdots\qquad \\[2mm]
+\, {\hbar^{k/2}}W_{a_1a_2a_3}\hat s^{a_1}
\hat s^{a_2}
E^{a_3}
+\hbar^{(k+2)/2} \, \alpha_{(k+2)/2}\, ,\quad k \mbox{ even}\, ,
\end{array}
\right.
$$
where the tensors $W$ are totally symmetric and $\alpha_{(k+2)/2}$ is some one-form. Both the~$W$'s and $\alpha_{(k+2)/2}$ are $\hbar$-independent.
Indeed, all  terms save the one-form $\alpha_{(k+2)/2}$
can---{\it mutatis mutandis}---be removed by higher order analogs of the gauge transformation $\exp(\hat u_1)$ employed in the base step above. Hence we have now proven that locally, gauge transformations achieve the form (formally to any power in the grading)
$$
\nabla=
\nabla_{\rm D} -i \sum_{j>1} \hbar^{j-1} \alpha_j\, .
$$
It only remains to apply the flatness condition one more time to show that the one-form $\alpha_\hbar :=\sum_{j>1} \hbar^{j-1} \alpha_j$ is 
closed and therefore locally $\alpha_\hbar = d \beta_\hbar$ 
for some function $\beta_\hbar$. Thus $\exp(i\beta_\hbar) \nabla \exp(-i\beta_\hbar)=\nabla_{\rm D}$.
\end{proof}

\begin{example}[The harmonic oscillator]
Let $Z={\mathbb R}^3=\{p,q,t\}$ and 
$$\alpha= pdq - \frac12(p^2+q^2)dt\, .$$
The Levi form
$$
\varphi = d\uppi\wedge d\upchi\,  ,
$$
where
$$
\uppi = \frac12 (p^2+q^2)\, ,\quad
\upchi = -t-\arctan(p/q)\, .
$$
Indeed, setting $\uppsi=-\frac12 pq$, we have $\alpha = \uppi d\upchi -d\uppsi$, so $(\uppi,\upchi,\uppsi)$ are local Darboux coordinates and (denoting $\hat s^a:=(\hat S,\hat P)$) the Darboux normal form~\nn{dbc} for the connection becomes
\begin{equation}\label{DB3}
\nabla_{\rm D}:=\frac{\uppi d\upchi -d\uppsi}{i\hbar}+\frac{\hat S d\uppi-\hat P d\upchi}{i\hbar} + d\
\end{equation}
Let us now run the steps of the above proof in {\it reverse} to show how to find gauge transformations bringing $\nabla_{\rm D}$ to the 
Hamiltonian dynamics form of~\ref{Hnabla}.

The closed soldering forms $E^a=(d\upchi,d\uppi)$ are related to those of the Hamiltonian dynamics Example~\ref{HamD} (given here by $e^a=(dq-pdt,dp+qdt)=:(f,e)$) according to  the $Sp(2)$ transformation
$$
E^a:=
\begin{pmatrix}
d\upchi
\\[1mm]  d\uppi
\end{pmatrix}=
\begin{pmatrix}
\frac{p}{2\uppi}
&
-\frac{q}{2\uppi} 
 \\[1mm]
q & p
\end{pmatrix}
\begin{pmatrix}
dq-pdt\\[1mm] dp+q dt
\end{pmatrix}=: U^a{}_b e^b\, .
$$
Writing $U=\exp(u)$ and then intertwining to its metaplectic representation
$\widehat U := \exp\big(
\frac{
J_{ac}u^c{}_b
\hat s^a \hat s^b}{2i\hbar}\big)$, 
we have $\widehat U^{-1} \big(\frac{\alpha}{i\hbar}+\frac{E^a J_{ab} \hat s^a}{i\hbar}\big)\widehat U=\frac{\alpha}{i\hbar}+\frac{e^a J_{ab} \hat s^a}{i\hbar}$, while a
short computation shows that the  $\mathfrak{sp}(2)$-valued one-form $U^{-1} d U$ is given explicitly by
$$
U^{-1}d U= \begin{pmatrix}
\, 0&-dt\, 
\\[3mm]\, dt&0\, 
\end{pmatrix}
+\begin{pmatrix}
-\frac{(p^2-q^2)(pe+qf)}{4\uppi^2}
&
\frac{(3p^2+q^2)q e-(p^2-q^2)pf}{4\uppi^2}
 \\[2mm]
 \frac{
(p^2-q^2)qe+(p^2+3q^2)pf}{4\uppi^2}
&
\frac{(p^2-q^2)(pe+qf)}{4\uppi^2}
\end{pmatrix}\, .
$$
It is not difficult to verify that the 
last term in the above display, which can be re-expressed as $W^a{}_{bc}e^c$ where the tensor $W_{abc}$ (moving indices with the antisymmetric bilinear form $J$) is totally symmetric\footnote{
Note that $$W_{222}=\frac{(3p^2+q^2)p}{4\uppi^2}\, , \quad W_{221}=-\frac{(p^2-q^2)q}{4\uppi^2}\, ,\quad
W_{211}=-\frac{(p^2-q^2)p}{4\uppi^2}\, ,\quad
W_{111}=\frac{(3p^2+q^2)q}{4\uppi^2}\, .$$}.
Moreover, interwining the first term to the metaplectic representation gives  the standard harmonic oscillator hamiltonian
$\frac i{2\hbar}\, dt (\hat P^2 + \hat S^2)$. Hence
the difference between the gauge transformed Darboux connection  and the Hamiltonian dynamics connection  of Equation~\ref{Hnabla} is 
$$
\widehat U^{-1} \nabla_{\rm D} \widehat U-\nabla=\frac{\hat s^a \hat s^b W_{abc}e^c}{2i\hbar}\, .
$$
The above term is order $0$ in the grading ${\sf gr}$ and therefore seeds the recursion described in the proof of Theorem~\ref{gaugeflat}. It is removed by a grade $1$ gauge transformation $\exp(\hat u_1)$ with $\hat u_1 = \frac{\hat s^a \hat s^b \hat s^cW_{abc}}{3!i\hbar}$.
It would desirable to have an efficient recursion to compute all higher terms with respect to the grading ${\sf gr}$ for the gauge transformation between $\nabla$ and $\nabla_{\rm D}$,
because in a general setting this would facilitate computation of quantum correlators.
\end{example}

\subsection{Contact deformation quantization}\label{Contactdef}
The above proof of gauge equivalence of flat connections is very close in spirit to Fedosov's formal quantization for symplectic and Poisson structures\footnote{Deformation quantization dates back to the seminal work of Bayen {\it et al}~\cite{Bayen}, see also~\cite{Beliavsky} for a review of symplectic connections.}. That work is concerned with constructing a quantum deformation of the Moyal star product, while here we wish to describe both dynamics and quantization. Nonetheless, we can employ's Fedosov's method to our quantized contact connection $\nabla$, to find a quantum deformation of the commutative algebra of classical solutions.

To study the algebra of operators, instead of the Hilbert bundle over $Z$,
we consider a {\it Heisenberg bundle}~$Z\ltimes {\mathcal U}({\mathfrak {heis}})$, defined in the same way as the Weyl bundle, except that instead of working with fibers given by functions of ${\mathbb R}^{2n}$ with a non-commutative Moyal star product, we work directly with operators\footnote{Recall that the Moyal star product amounts simply to coordinatizing the space of operators ${\mathcal U}({\mathfrak {heis}})$ in terms of functions of ${\mathbb R}^{2n}$ by employing a Weyl-ordered operator basis, and then encoding their algebra using a non-commutative $\star$-multiplication of functions.}. For our purposes, the key point is that local sections $\hat a$ of the Heisenberg bundle are functions of $Z$ taking values in  ${\mathcal U}({\mathfrak {heis}})$, which can be expressed with respect to the grading ${\sf gr}$ as  
$$
\hat a=\underbrace{\frac{a^{\sss(-2)}}{i\hbar} }_{-2}+ \underbrace{\frac{a^{\sss(-1)}_a\hat s^a}{i\hbar}}_{-1} + 
\underbrace{\frac{a^{\sss(0)}_{ab}\hat s^a\hat s^b}{2i\hbar}-ia^{\sss(0)}}_{0}+\cdots
$$
Importantly, $a^{(k)}$ are $\hbar$ independent, and we do not allow negative powers of $\hbar$ greater than one.

Requiring total symmetry of the tensors $a^{(k)}_{a_1\ldots a_{j\leq k}}$ appearing in the above expansion uniquely determines a function of $\hbar$ which---following Fedosov---we call the {\it abelian part} of $\hat a$ and denote by
$$
\sigma(\hat a) := a^{\sss(-2)}+\hbar a^{\sss(0)} + \hbar^2 a^{\sss(2)} + \cdots\, . 
$$
We call $\hat a-\frac1 {i\hbar}\sigma(\hat a)$ the {\it non-abelian} part of $\hat a$.

The flat 
connection $\nabla$ acts on sections of 
the Heisenberg bundle by the adjoint action
$$
\nabla \hat a := d\hat a + [\widehat A,\hat a]\, .
$$
The following lemma locally characterizes parallel sections.

\begin{lemma}
Let $f_\hbar\in C^\infty Z[[\hbar]]$ obey
$$
{\mathcal L}_\rho f_\hbar = 0\, .
$$
Then locally, there is a unique section $\hat a\in \Gamma(Z\ltimes {\mathcal U}({\mathfrak {heis}}))$ such that
$$
\nabla \hat a = 0 \mbox{ and } \sigma(\hat a) = f_\hbar\, .
$$
\end{lemma}


\begin{proof}
By virtue of Theorem~\ref{gaugeflat}
we know that locally 
$$\nabla= \exp(\hat u) \circ \nabla_{\rm D} \circ \exp(-\hat u)\, ,$$
for some $\hat u\in \Gamma(Z\ltimes {\mathcal U}({\mathfrak {heis}}))$ and $\nabla_{\rm D}$ is given by Equation~\nn{dbc}. Therefore we begin by establishing that the equation 
\begin{equation}\label{pds}
\nabla_{\rm D} \hat b=0
\end{equation} 
has a solution
such that 
\begin{equation}
\label{boundary}
\sigma(\exp(\hat u) \, \hat b \, \exp(-\hat u))=f_\hbar\, ,
\end{equation}
because $\hat a = \exp(\hat u)$ $\hat b \exp(-\hat u)$  will then solve $\nabla \hat a= 0$ with the correct boundary condition $\sigma(\hat a)=f_\hbar$. (We deal with uniqueness at the end of this proof.)

We now work order by order in the grading ${\sf gr}$. Firstly, we must solve 
$$
0={\sf gr}_{-2}(\nabla_{\rm D} \hat b)=
\frac{db^{\sss(-2)}+ b^{\sss(-1)}_a E^a}{i\hbar}\, .
$$
From Equation~\nn{boundary} we have $b^{\sss(-2)}=a^{\sss(-2)}={\sf gr}_{-2} f_\hbar$, but by assumption ${\mathcal L}_\rho f_\hbar=0$ so Cartan's magic lemma gives
$\iota_\rho db^{\sss(-2)} = 0$,
whence $db^{\sss(-2)}\in {\rm span}\{E^a\}$. Hence we can solve the equation in the above display (uniquely) for $b^{\sss(-1)}_a$.
 
At the next order in the grading we must now solve
$$
0={\sf gr}_{-1}(\nabla_{\rm D} \hat b)=
\frac{db^{\sss(-1)}_a\hat s^a+ b^{\sss(0)}_{ab} E^a\hat s^b}{i\hbar}\, .
$$
By virtue of the Darboux coordinate system, $b_a^{\sss (-1)}$ cannot depend on $\uppsi$ so $\iota_\rho db_a^{\sss(-1)}=0$.
Hence the above display (uniquely) determines $b^{\sss(0)}_{ab}$ (and once again $\iota_\rho db_{ab}^{\sss(0)}=0$). The abelian term $-ib^{\sss(0)}$ is at this point not determined. However for that we
impose Equation~\nn{boundary} to the order $0$ in the grading, which now
determines $ b^{\sss(0)}$ in terms of $f_\hbar$ and other $\uppsi$-independent quantities. This establishes the pattern for an obvious recursion, which completes the existence part of this proof.

%
%

To show uniqueness, suppose $\hat a'$ also obeys $\nabla \hat a'=0$ such that $\sigma(\hat a'-\hat a)=0$. Now, let 
$$
\nabla = \frac{\alpha}{i\hbar}+ \frac{e^aJ_{ab}\hat s^b}{i\hbar} + \cdots\, .
$$
Then 
$$0={\sf gr}_{-2}\big(\nabla(\hat a'-\hat a)\big) = \frac{(a'^{\sss(-1)}_a-a^{\sss(-1)}_a)e^a}{i\hbar}\:\Leftrightarrow\:
a'^{\sss(-1)}_a=a^{\sss(-1)}_a\, .
$$
Indeed, the same pattern holds at all higher orders in the grading ${\sf gr}$, so that $\hat a'=\hat a$, as required. 
\end{proof}

\begin{remark}
Calling $\upxi^a=(\upchi^i,\uppi_i)$, the Darboux connection~\nn{dbc} obeys
$$
[\nabla_{\!\rm D},\hat s^a - \upxi^a]=0\, .
$$
So taking $\hat b$ equal to any polynomial ${\mathcal P}(\hat s^a - \upxi^a)$ solves the parallel section condition~\nn{pds}.
This in turn immediately solves the parallel section problem for $f_\hbar$ expressible as polynomial in Darboux coordinates.
Note however, that in general, 
replacing ${\mathcal P}$ by a formal power series in $\hat s^a-\upxi^a$, may not give a well defined formal power series
in Weyl ordered symbols of~$\hat s^a$.
(Quantum reordering terms potentially involve infinite, non-convergent, sums of the coefficients of the original power series.)
\end{remark}

\medskip

Let us denote by $\sigma^{-1}$ the map $C^\infty Z[[\hbar]]\cap{\rm ker}({\mathcal L}_\rho)\ni f_\hbar \mapsto \hat a$ as
defined by the above lemma.
Now consider a pair of solutions $f_\hbar,g_\hbar\in C^\infty Z[[\hbar]]$ to the classical equations of motion: 
$$
{\mathcal L}_\rho f_\hbar = 0 = {\mathcal L}_\rho g_\hbar\, .
$$
Then we have a pair of parallel sections $\sigma^{-1}(f_\hbar)$ and $\sigma^{-1}(g_\hbar)$ of $Z\ltimes {\mathcal U}(\frak{heis})$. These may be multiplied pointwise along $Z$ using the  operator product on fibers. 
Therefore, {\it a l\'a} Fedosov~\cite{Fed},  we may define a $\star$-multiplication of functions $f_\hbar$ and $g_\hbar$ by\footnote{Fedosov constructs a deformation of  the Moyal star product for Weyl ordered operators
in the Weyl algebra given the data of a symplectic manifold. Here we skip the Moyal star  and work directly with operators in the Weyl algebra.
}  
$$
f_\hbar \star g_\hbar = \sigma\big(\sigma^{-1}(f_\hbar)
 \sigma^{-1}(g_\hbar)\big)\, .
$$
This gives a contact analog of deformation quantization. Observe that it reduces the deformation problem to a gauge transformation. However, unlike Fedosov's work, this means that the above uniqueness proof for flat sections is local. It ought however be possible to improve this to a global statement and preliminary results indicate that this is the case; we reserve those results for a later publication, where we also plan to detail the  precise map between the above display and Fedosov's deformation formula for symplectic structures.

\section{Flat Sections and Dynamics}\label{dynamics}

As discussed in the previous section, solving for a flat connection $\nabla$ on the Hilbert bundle $Z\ltimes {\mathcal H}$ 
is analogous to finding an operator quantizing a classical Hamiltonian, while the parallel transport equation~\nn{Schroedinger} is the analog of the Schr\"odinger equation which controls quantum dynamics.
 We now turn our attention to solving the latter and computing correlators.
 
 \medskip
 
 \subsection{Parallel transport}
 
 Let us suppose we have prepared a state $|{\mathcal E}_{\rm i}\rangle \in {\mathcal H}_{z_{\rm i}}$ where~${\mathcal H}_{z_{\rm i}}$ is the Hilbert space associated with a point $z_{\rm i}\in Z$
 (one may think of $z\in Z$ as a generalized laboratory time coordinate). We would like to compute the probability of measuring
 a state  $|{\mathcal E}_{\rm f}\rangle \in {\mathcal H}_{z_{\rm f}}$ at some other 
 point $z_{\rm f}\in Z$. For that, observe that we can  parallel transport the ``initial'' state $|{\mathcal E}_{\rm i}\rangle$
 from the Hilbert space ${\mathcal H}_{z_{\rm i}}$ to any other
 Hilbert space ${\mathcal H}_{z}$
using a line operator
\begin{equation}
\label{lineop}
|{\mathcal E}(z)\rangle=\Big({\rm P}_\gamma
\exp\big(\!-\!\int_{z_{\rm i}}^z\!\widehat A\,\big)\Big)
|{\mathcal E}_{{\rm i}}\rangle \in {\mathcal H}_{z}\, ,
\end{equation}
where ${\rm P}_\gamma$ denotes path ordering and $\gamma$ is any path in $Z$ joining $z_{\rm i}$ and $z$. Since $\nabla = d+\hat A$, it follows that the section $\Psi(z)=|{\mathcal E}(z)\rangle$ of $Z\ltimes {\mathcal H}$ solves the Schr\"odinger equation~\nn{Schroedinger}.
Since the connection $\nabla$ is flat, if the fundamental group $\pi_1(Z)$ is trivial, this solution is independent of the choice of path $\gamma$ between $z_{\rm i}$ and $z$. When this is not the case, we must be more careful with the choice of Hilbert space fibers. We discuss this further below. Modulo this issue, the probability $P_{\rm f,i}$ of observing $|{\mathcal E}_{\rm f}\rangle \in {\mathcal H}_{z_{\rm f}}$ having prepared 
$|{\mathcal E}_{\rm i}\rangle \in {\mathcal H}_{z_{\rm i}}$ is
$$P_{\rm f,i}
=
\frac{\Big|
\langle {\mathcal E}_{\rm f}|
\Big({\rm P}_\gamma
\exp\big(\!-\!\int_{z_{\rm i}}^{z_{\rm f}}\!\widehat A\,\big)\Big)
|{\mathcal E}_{{\rm i}}\rangle\Big|^2}
{\langle {\mathcal E}_{\rm f}|
{\mathcal E}_{\rm f}\rangle
\, 
\langle {\mathcal E}_{\rm i}|
{\mathcal E}_{\rm i}\rangle
}\, .
$$
In~\cite{Herczeg} we showed how to extract  quantum mechanical  Wigner functions from correlators
\begin{equation}
\label{correlator}{\mathcal W}_{{\mathcal E}_{\rm f},{\mathcal E}_{\rm i}}(z_{\rm f},z_{\rm i}):=\langle {\mathcal E}_{\rm f}|
\Big({\rm P}_\gamma
\exp\big(\!-\!\int_{z_{\rm i}}^{z_{\rm f}}\!\widehat A\,\big)\Big)
|{\mathcal E}_{{\rm i}}\rangle\, .
\end{equation}
This correlator 
is gauge covariant. In particular, 
in a contractible local patch around the path $\gamma$, by virtue of Theorem~\ref{gaugeflat}, we can find a gauge transformation $\widehat U$ such that $\widehat U \nabla \widehat U^{-1}=\nabla_{\rm D}$, where the Darboux 
normal form is given in Equation~\nn{normal}. Hence the line operators for these two connections are 
related by
\begin{equation}\label{gcov}
\Big({\rm P}_\gamma
\exp\big(\!-\!\int_{z_{\rm i}}^{z_{\rm f}}\!\widehat A\,\big)\Big)=\widehat U(z_{\rm f})^{-1}\circ
\Big({\rm P}_\gamma
\exp\big(\!-\!\int_{z_{\rm i}}^z\!\widehat A_{\rm D}\,\big)\Big)\circ
\widehat U(z_{\rm i})\, .
\end{equation} 
 Inserting resolutions 
of unity  $\int dS |S\rangle\langle S|=1=\int dP |P\rangle\langle P|$ for $\mathcal H$
(where $\hat s^a=(\hat S^A,\hat P_A)$ and $\hat S_A |S\rangle =S_A |S\rangle$, $\hat P_A |P\rangle =P_A |P\rangle$)
in the  above identity, and putting this in the correlator~\nn{correlator}
gives~\footnote{\label{foot}Of course, one could equally well insert other resolutions of unity, for example, replacing $\int dP |P\rangle\langle P|$ with $\int dS' |S'\rangle\langle S'|$ is a propitious choice used in the next example.}
\begin{equation}\label{resolve}
{\mathcal W}_{{\mathcal E}_{\rm f},{\mathcal E}_{\rm i}}(z_{\rm f},z_{\rm i}):=\int dS dP\,  \langle  {\mathcal E}_{\rm f}|\hh\widehat U(z_{\rm f})^{-1}|P\rangle\, 
\langle P |
\Big({\rm P}_\gamma
\exp\big(\!-\!\int_{z_{\rm i}}^{z_{\rm f}}\!\widehat A_{\rm D}\,\big)\Big) |S\rangle\, 
\langle S|\hh \widehat U(z_{\rm i})
|{\mathcal E}_{{\rm i}}\rangle\, .
\end{equation}
Since the line operator for the connection $\widehat A_{\rm D}$ in the Darboux frame is essentially trivial (see directly below), knowledge of the gauge transformations $\widehat U$ 
determines the correlator. 

\begin{example}[The Darboux correlator]
Consider a pair of points $z_{\rm i}=(\uppi_{\rm i},\upchi_{\rm i},\uppsi_{\rm i})$ and $z_{\rm f}=(\uppi_{\rm f},\upchi_{\rm f},\uppsi_{\rm f})$ in the contact three-manifold $Z=({\mathbb R}^3,\uppi d\upchi-d\uppsi)$.
Since here we want to  study a line operator for a  flat connection $\nabla_{\rm D}$ on a trivial manifold, we may choose any path between these two points, so take 
$\gamma=  \gamma_\uppsi\cup\gamma_\uppi\cup\gamma_\upchi$
where
\begin{eqnarray*}
\gamma_\uppi&:=&\{
(1-t)\uppi_{\rm i}+t\uppi_{\rm f},\upchi_{\rm i},\uppsi_{\rm i})
\}\, ,
\\[1mm]
\gamma_\upchi&:=&
\{(
\uppi_{\rm f},(1-t)\upchi_{\rm i}+t\upchi_{\rm f},\uppsi_{\rm i})\}
\, ,\\[1mm]
\gamma_\uppsi&:=&
\{
(
\uppi_{\rm f},\upchi_{\rm f},(1-t)\uppsi_{\rm i}+t\uppsi_{\rm f})\}\, ,
\\[1mm]
\end{eqnarray*}
where $t\in[0,1]$.
Then, along these three paths the potential $\widehat A$ for the
Darboux connection (see Equation~\ref{DB3}) takes the form
$$
\widehat A_{\gamma_\uppi}=\frac1{i\hbar}\, dt(\uppi_{\rm f}-\uppi_{\rm i}) \hat S\, ,\quad
\widehat A_{\gamma_\upchi}=\frac1{i\hbar}\, dt(\upchi_{\rm f}-\upchi_{\rm i}) (\uppi_{\rm f}-\hat P)\, ,\quad
\widehat A_{\gamma_\uppsi}=-\frac1{i\hbar}\, dt(\uppsi_{\rm f}-\uppsi_{\rm i}) \, .
$$
Hence the correlator in Darboux frame is simply
$$
\langle P |
\Big({\rm P}_\gamma
\exp\big(\!-\!\int_{z_{\rm i}}^{z_{\rm f}}\!\widehat A_{\rm D}\,\big)\Big) |S\rangle=
\exp\Big(-
\frac{
(\upchi_{\rm f}-\upchi_{\rm i}) (\uppi_{\rm f}- P)+
(\uppi_{\rm f}-\uppi_{\rm i}) S
-\uppsi_{\rm f}+\uppsi_{\rm i}}
{i\hbar}
\Big)\, .
$$
The above result combined with Equation~\ref{resolve} indeed shows that  knowledge of the gauge transformation $\widehat U$ bringing a connection to its Darboux form determines correlators.
\end{example}

\subsection{Path integrals}

In general, one does not have access to the explicit diffeomorphism  bringing the contact form to its Darboux normal form (let alone the gauge transformation~$\widehat U$). Instead  correlators can be computed
 in terms of path integrals.
For that, per its definition, we split the path ordered exponential of the integrated potential $\widehat A$ 
into infinitesimal segments $dz^i$ along the path $\gamma$, and insert successive resolutions of unity. In particular, 
using that, for $dz^i$ small,
$$
\langle P | \exp(- \widehat A_i(\hat S,\hat P) dz^i)|S\rangle
\approx\exp\big(\frac i\hbar P_A S^A - A_{\rm N}(S,P)\big)\, ,
$$
where $A_{\rm N}(S,P)$ is the normal ordered symbol\footnote{To be precise, $\widehat A$ is recovered by writing $A(S,P)$ as a power series in $P$ and $S$ and then replacing monomials $P^k S^l$ by the operator $\hat P^k \hat S^l$.} of the operator $\hat A$,
we have the operator relation
$$
\exp(- \widehat A_i dz^i) \approx \int dS dP\,  |P\rangle
\exp\big(\frac i\hbar P_A S^A - A_{\rm N}(S,P)\big) \langle S | \, .
$$
Concatentating this expression along the path $\gamma$ gives the path integral formula for the correlator between states $|S_{\rm i}\rangle$ and $\langle P_{\rm f}|$
$$
{\mathcal W}_{P_{\rm f},S_{\rm i}}(z_{\rm f},z_{\rm i})=\int_{S(z_{\rm i})=S_{\rm i}}
^{P(z_{\rm f})=S_{\rm f}}
 [dP dS] \exp\Big(-
\frac{i}{\hbar} \int_\gamma \big(P_A dS^A +A_{\rm N}(S,P)\big)\Big)\, .
$$
In the above $\gamma$ is any path in $Z$ connecting $z_{\rm i}$ and $z_{\rm f}$. When $\nabla$ has trivial holonomy (otherwise see below), neither the correlator nor its path integral representation depends on this choice.
Notice that the path integration in the above formula is only performed fiberwise. We do not integrate over paths $\gamma$ in~$Z$, but rather paths in the total space ${\mathcal Z}=Z\ltimes{\mathbb R}^{2n}$ above the path $\gamma$ in $Z$. Indeed, calling $s^a:=(S^A,P_A)$ 
and writing $P_A dS^A = \frac 12 s^a J_{ab}  ds^b$ we see that the action appearing in the exponent of the above path integral is 
the quantum corrected analog of the   extended action of Equation~\nn{extS}(computing the operator~$\widehat A$ and its normal ordered symbol $A_{\rm N}$ will in general produce terms proportional to powers of~$\hbar$).

\subsection{Topology}

Finally, we discuss the case when the fundamental group~$\pi_1(Z)$ is non trivial\footnote{We owe the key idea of this section of modding out the Hilbert space fibers by the holonomy of $\nabla$  to Tudor Dimofte.}. The holonomy of the connection $\nabla$
may then be non-trivial, and the parallel transport solution~\nn{lineop}
to the Schr\"odinger equation can depend on the homotopy class of the path $\gamma$. {\it A priori} this seems to be a {\it bug} leading to loss of predictivity, however remembering that the topology of system can influence its quantum spectrum (consider a free particle in a box, for example), we have in fact hit upon a {\it feature}. Our quantization procedure is not complete until we impose that  the holonomy of the connection~$\nabla$ acts trivially on the Hilbert space fibers.
To explain this point better, as a running example consider the contact form
$$
\alpha = \uppi d \uptheta - d\uppsi\, ,
$$
on the manifold $Z=C\times {\mathbb R}$ where $C$ is a cylinder with periodic coordinate $\uptheta\sim \uptheta+2\pi$. Now let us study the quantizaton determined by the flat connection $\nabla = d+\widehat A$ where
$$
\widehat A = \frac{\alpha}{i\hbar} +
d \uppi \frac{S}{i\hbar}+
d\uptheta\Partial{S} 
 \, .
$$
Here we have picked some polarization for the Hilbert space fibers such that elements are given by wavefunctions $\psi(S)$.

Along the path $\gamma=\{\uptheta=\uptheta_o+\theta,\uppi=\uppi_o,\uppsi=\uppsi_o\, :\, \theta\in[0,2\pi)\}$, we have $\widehat A_\gamma
=\frac{1}{i\hbar}d\theta\big(\uppi_o -\frac \hbar i \Partial S\big)$. Hence the holonomy of $\nabla$ at basepoint $z_o=(\uptheta_o,\uppi_o,\uppsi_o)$ is
$$
{\rm hol}_{z_o}(\widehat A_\gamma) = \exp\Big(-\frac{2\pi i}\hbar\big(\uppi_o
-\frac \hbar i \Partial S\big)\Big)\, .
$$
Requiring that this holonomy 
acts trivially on the Hilbert space ${\mathcal H}$ over the base point $z_o\in Z$, we impose that elements $\psi_{z_o}(S)$ of that space obey
$$
\exp\Big(-\frac{2\pi i}\hbar\big(\uppi_o
-\frac \hbar i \Partial S\big)\Big)\, \psi_{z_o}(S)=\psi_{z_o}(S)\, .
$$
Hence
$$
\psi_{z_o}(S+2\pi)=e^{\frac{2\pi i\uppi_o}{\hbar}} \psi_{z_o}(S)\, .
$$
So, up to  a basepoint dependent phase, wavefunctions are periodic.
In effect, the classical topology of the contact base manifold $Z$ has enforced the desired boundary conditions on quantum wavefunctions.

\section{Discussion and Conclusions}\label{discuss}

Just as contact geometry reduces classical mechanics to a problem of contact topology (all dynamics is locally trivial by virtue of the contact Darboux theorem), the contact quantization we have presented does the same for quantum dynamics. Moreover, since our approach is completely generally covariant, even seemingly disparate systems can be related by appropriate choices of clocks. This gives a concrete setting for quantum cosmology-motivated studies of 
 the  ``clock ambiguity'' of  quantum dynamics~\cite{Albrecht,Warsaw}.

\medskip

Beyond providing a solid mathematical framework for philosophical questions of time and measurement in quantum mechanics, it is very interesting to 
probe to which extent the gauge
freedom characterized in Theorem~\ref{gaugeflat} can be used to solve or  further the study of concrete 
quantum mechanical systems. As discussed in Section~\ref{dynamics}, knowledge of the gauge transformation bringing the connection $\nabla$ to its Darboux form can be used to compute correlators, which begs the question whether methods---perturbative, exact when symmetries are present, or numerical---can be developed to calculate these transformations.

Along similar lines to the above remark, symmetries and integrability play a central {\it r\^ole} in the analysis
of quantum systems. Again contact geometry and its quantization ought be an ideal setting for analyzing quantum symmetries and relating them to contact topology. Preliminary results show that this is case, and we plan to report on such questions elsewhere.

\medskip

Lattice spin models and models with Fermi statistics are crucial for the description of physical systems. Here one needs to study supercontact structures (see~\cite{Manin, Schwarz, Bruce}); it is indeed not difficult to verify that our flat connection/quantizaton and parallel section/dynamics methodology can be applied directly in the supercontact setting; again we plan to report on this interesting direction in the near future. 

\medskip

In Section~\ref{Contactdef}
we showed how to relate contact quantization  to Fedosov's deformation quantization.
It would also be interesting to relate our approach to other quantization methods. In particular, it would be interesting to study the relation to Kontsevich's explicit deformation quantization formula for Poisson structures~\cite{Kontsevich} and its Cattaneo--Felder sigma model derivation~\cite{Cattaneo}. In addition, it would be interesting to study when we can go beyond formal deformation  quantization, perhaps along the lines of  the $A$-model approach of Gukov--Witten to quantization~\cite{Gukov}, or geometric quantization in general. Indeed, Fitzpatrick has made a  rigorous geometric quantization study of contact structures~\cite{Fitzpatrick} based on the proposal by Rajeev~\cite{Rajeev} to quantize Lagrange brackets (these are the contact analog of the Poisson bracket). Note also that earlier work by Kashiwara~\cite{Kashi}  studies sheaves of pseudodifferential operators over contact manifolds, and Yoshioka has performed a contact analog of Fedosov quantization where the base manifold is a symplectic manifold and the fibers carry a contact structure~\cite{Yoshioka}.

Finally, we mention that our construction of 
the connection~$\nabla$ is in spirit rather close to the Cartan normal connection in parabolic geometries, see~\cite{Slovak} for the general theory and~\cite{Fox} for its application to contact structures compatible with a projective structure. These geometric methods may also end up being directly relevant to quantum mechanics.

\section*{Acknowledgements}
This work was presented in part in lectures and further developed at the 38th Geometry and Physics Winter School in Srn\'i. A.W. and E.L. thank the organizers for this wonderful forum for interactions between geometry and physicists.
We also thank Andy Albrecht, Roberto Bonezzi,  Steve Carlip, James Conway, Olindo Corradini, Tudor Dimofte, Mike Eastwood, Rod Gover, Maxim Grigoriev, Jerry Kaminker,  Bruno Nachtergaele and Andrea Santi for discussions.
A.W. was supported in part by a 
Simons Foundation Collaboration Grant for Mathematicians ID 317562.

\end{document}